\documentclass{elsarticle}
\usepackage{amsmath,amsfonts,amsthm}
\usepackage{amssymb,latexsym}

\theoremstyle{plain}
\newtheorem{theorem}{Theorem}[section]
\newtheorem{lemma}[theorem]{Lemma}

\theoremstyle{definition}
\newtheorem{mydefn}[theorem]{Definition}

\newtheorem{observation}[theorem]{Observation}

\usepackage{mathrsfs}
\usepackage{amsfonts}
\usepackage{graphicx}
\usepackage{slashed}

\newcommand{\action}[1]{\stackrel{#1}{\rightarrow}}
\newcommand{\faction}[1]{\stackrel{#1}{\rightarrow}_F}
\newcommand{\actions}[1]{\stackrel{#1}{\Rightarrow}}
\newcommand{\sactions}[1]{\stackrel{#1}{\Rightarrow}|}
\newcommand{\fsactions}[1]{\stackrel{#1}{\Rightarrow}_F|}
\newcommand{\SRS}{\underset{\thicksim}{\sqsubset}_{RS}}
\newcommand{\RS}{\sqsubseteq_{RS}}
\newcommand{\everyR}[1]{{for each $#1 \in \widetilde{#1}$}}
\newcommand{\they}{{L\"{u}ttgen and Vogler}}
\begin{document}

\author[nuaa]{Yan Zhang}

\author[nuaa]{Zhaohui Zhu\corref{cor}}

\author[nanshen]{Jinjin Zhang}

\cortext[cor]{Corresponding author. Email: zhaohui@nuaa.edu.cn, zhaohui.nuaa@gmail.com (Zhaohui Zhu).}

\address[nuaa]{College of Computer Science, Nanjing University of Aeronautics and Astronautics, Nanjing, P.R. China, 210016}
\address[nanshen]{College of Information Science, Nanjing Audit University, \\Nanjing, P.R. China, 211815}

\title{Greatest solutions of equations in $\text{CLL}_R$ and its application\tnoteref{t1}}
\tnotetext[t1]{This work received financial support of the National Natural Science of China (No. 60973045) and Fok Ying-Tung Education Foundation, NSF of the Jiangsu Higher Education Institutions (No. 13KJB520012) and the Fundamental Research Funds for the Central Universities (No. NZ2013306).}
\date{\today}
\begin{abstract}
 This paper explores the process calculus $\text{CLL}_R$ furtherly.
 First, we prove that for any equation $X=_{RS} t_X$ such that $X$ is strongly guarded in $t_X$, $\langle X|X=t_X \rangle$ is the largest solution w.r.t $\RS$.
 Second, we encode a fragment of action-based CTL in $\text{CLL}_R$.
\end{abstract}



\maketitle
%

\section{Introduction}

    It is well-known that process algebra and temporal logic take different standpoint for looking at specifications and verifications of reactive and concurrent systems, and offer complementary advantages \cite{Peled01}.
    To take advantage of these two paradigms when designing systems, a few of theories for heterogeneous specifications have been proposed, e.g., \cite{Cleaveland00, Cleaveland02, Graf86,Kurshan94,Luttgen07,Luttgen10,Luttgen11,Olderog}.
    Among them, L{\"u}ttgen and Vogler propose the notion of logic labelled transition system (Logic LTS or LLTS for short), which combines operational and logical styles of specification in one unified framework \cite{Luttgen07,Luttgen10,Luttgen11}.
    In addition to usual process operators (e.g., CSP-style parallel composition, hiding, etc) and logic operators (disjunction and conjunction), some standard temporal logic operators, such as \textquotedblleft always\textquotedblright and \textquotedblleft unless\textquotedblright, are also integrated into this framework \cite{Luttgen11}, which allows ones to freely mix operational and logic operators when designing systems.

   L\"{u}ttgen and Vogler's approach is entirely semantic, and doesn't provide any kind of syntactic calculus.
   Recently, we propose a LLTS-oriented process calculus $\text{CLL}_R$, and establish the uniqueness of solutions of equations in $\text{CLL}_R$ under a certain circumstance \cite{Zhang14}.

 This paper works on $\text{CLL}_R$ furtherly. Our main contributions include:

(1) We will show that, without the assumption that $X$ does ont occur in the scope of any conjunction in $t$, the given equation $X=_{RS}t$ may have more than one consistent solution.
       This answers conjecture in \cite{Zhang14} negatively.
       Under the hypothesis that $X$ is strongly guarded in a given (open) term $t$, it is shown that the recursive process $\langle X|X=t\rangle$ is indeed the greatest (w.r.t $\sqsubseteq_{RS}$) consistent solution of the equation $X=_{RS}t$ whenever consistent solutions exist.

(2) We encode a temporal logic language action-based CTL  \cite{Luttgen11} in $\text{CLL}_R$ so that safety properties could be described directly without resorting to complicated settings \cite{Luttgen11}, which are used to embed temporal logic operators into LLTS.

    The rest of this paper is organized as follows.
    The calculus $\text{CLL}_R$ and its semantics are recalled in the next section.
    Section~3 show that for any given equation $X=_{RS}t$ such that $X$ is strongly guarded in $t$, $\langle X |X =t \rangle$ is the greatest solution w.r.t $\RS$.
    We encode action-based CTL in Section~4.
    The paper is concluded with Section~5, where a brief discussion is given.

\section{Preliminaries}

The purpose of this section is to fix our notation and terminology, and to introduce some concepts that underlie our work in all other parts of the paper.

\subsection{Logic LTS and ready simulation}


Let $Act$ be the set of visible action names ranged over by $a$, $b$, etc., and let $Act_{\tau}$ denote $Act \cup \{\tau\}$ ranged over by $\alpha$ and $\beta$, where $\tau$ represents invisible actions.
A labelled transition system with predicate is a quadruple $(P,Act_{\tau},\rightarrow, F)$, where $P$ is a set of states, $\rightarrow \subseteq P\times Act_{\tau}\times P$ is the transition relation and $F\subseteq P$.

As usual, we write $p \stackrel{\alpha}{\rightarrow}$ (or, $p \not \stackrel{\alpha}{\rightarrow}$) if $\exists q\in P.p\stackrel{\alpha}{\rightarrow}q$ ($\nexists q\in P.p  \stackrel{\alpha}{\rightarrow}q$, resp.).
The ready set $\{\alpha \in Act_{\tau}|p \stackrel{\alpha}{\rightarrow}\}$ of a given state $p$ is denoted by $\mathcal{I}(p)$.
A state $p$ is stable if $p \not\stackrel{\tau}{\rightarrow}$.
We also list some useful decorated transition relations:

$p \stackrel{\alpha}{\rightarrow}_F q$ iff $p \stackrel{\alpha}{\rightarrow} q$ and $p,q\notin F$;

$p \stackrel{\epsilon}{\Rightarrow}q$ iff $p (\stackrel{\tau}{\rightarrow})^* q$, where $(\stackrel{\tau}{\rightarrow})^* $ is the transitive and reflexive closure of $\stackrel{\tau}{\rightarrow}$;

$p \stackrel{\alpha}{\Rightarrow}q$ iff $\exists r,s\in P.p \stackrel{\epsilon}{\Rightarrow} r \stackrel{\alpha}{\rightarrow}s \stackrel{\epsilon}{\Rightarrow} q$;

$p \stackrel{\gamma}{\Rightarrow}|q$ iff $p \stackrel{\gamma}{\Rightarrow}q \not\stackrel{\tau}{\rightarrow}$ with $\gamma \in Act_{\tau}\cup \{\epsilon\}$;

$p\stackrel{\epsilon }{\Rightarrow }_Fq$ iff there exists a sequence of $\tau$-transitions from $p$ to $q$ such that all states along this sequence, including $p$ and $q$, are not in $F$; the decorated transition $p \stackrel{\alpha }{\Rightarrow }_Fq$ may be defined similarly;

$p \stackrel{\gamma}{\Rightarrow}_F|q$ iff $p \stackrel{\gamma}{\Rightarrow}_F q \not\stackrel{\tau}{\rightarrow}$ with $\gamma \in Act_{\tau} \cup \{\epsilon\}$.

 Notice that
the notation $p
\stackrel{\gamma }{\Longrightarrow }\mspace{-8mu}|q$ in \cite{Luttgen10,Luttgen11}
has the same meaning as $p\stackrel{\gamma }{
\Rightarrow }_F|q$ in this paper, while $p\stackrel{\gamma }{\Rightarrow }|q $ in this paper does not involve any requirement on $F$-predicate.

\begin{mydefn}[Logic LTS \cite{Luttgen10}]\label{D:LLTS}
 An LTS $(P,Act_{\tau},\rightarrow,F)$ is an LLTS if, for each $p \in P$,

\noindent\textbf{(LTS1) }$p \in F$ if $\exists\alpha\in \mathcal{I}(p)\forall q\in P(p \stackrel{\alpha}{\rightarrow}q \;\text{implies}\; q\in F)$;

\noindent\textbf{(LTS2)} $p \in F$ if $\nexists q\in P.p \stackrel{\epsilon}{\Rightarrow}_F|q$.

Moreover, an LTS $(P,Act_{\tau},\rightarrow,F)$ is {$\tau$}-pure if, for each $p \in P$, $p\stackrel{\tau}{\rightarrow}$ implies $\nexists a\in Act.\;p\stackrel{a}{\rightarrow}$.
\end{mydefn}

Compared with usual LTSs, it is one distinguishing feature of LLTS that it
involves consideration of inconsistencies.
The main motivation
behind such consideration lies in dealing with inconsistencies caused by
conjunctive composition.
Formally, the predicate $F$ is used to denote the set of all inconsistent states that represent empty behaviour that cannot be implemented \cite{Luttgen11}.
The condition (LTS1) formalizes the backward propagation of inconsistencies, and (LTS2) captures
the intuition that divergence (i.e., infinite sequences of $\tau $%
-transitions) should be viewed as catastrophic.
For more intuitive ideas and motivation about inconsistency, the reader may refer \cite{Luttgen07,Luttgen10}.

A variant of the usual notion of weak ready simulation \cite{Bloom95,Larsen91} is adopted to capture the refinement relation in \cite{Luttgen10,Luttgen11}.
It has been proven that such kind of ready simulation is the largest precongruence w.r.t parallel composition and conjunction which satisfies  the desired property that  an inconsistent specification can only be refined by inconsistent ones (see Theorem 21 in \cite{Luttgen10}).

\begin{mydefn}[Ready simulation on LLTS \cite{Luttgen10}]\label{D:RS}
Let $(P, Act_{\tau}, \rightarrow , F)$ be a  LLTS.
A relation ${\mathcal R} \subseteq P \times P$ is a stable ready simulation relation, if for any $(p,q) \in {\mathcal R}$ and $a \in Act $\\
\textbf{(RS1)} both $p$ and $q$ are stable;\\
\textbf{(RS2)} $p \notin F$ implies $q \notin F$;\\
\textbf{(RS3)} $p \stackrel{a}{\Rightarrow}_F|p'$ implies $\exists q'.q \stackrel{a}{\Rightarrow}_F|q'\; \textrm{and}\;(p',q') \in {\mathcal R}$;\\
\textbf{(RS4)} $p\notin F$ implies ${\mathcal I}(p)={\mathcal I}(q)$.

\noindent We say that $p$ is stable ready simulated by $q$, in symbols $p \underset{\thicksim}{\sqsubset}_{RS} q$, if there exists a stable ready simulation relation $\mathcal R$ with $(p,q) \in {\mathcal R}$.
 Further, $p$ is ready simulated by $q$, written $p\sqsubseteq_{RS}q$, if
 $\forall p'(p\stackrel{\epsilon}{\Rightarrow}_F| p' \;\text{implies}\; \exists q'(q \stackrel{\epsilon}{\Rightarrow}_F| q'\; \text{and}\;p' \underset{\thicksim}{\sqsubset}_{RS} q'))$.
 The kernels of $\underset{\thicksim}{\sqsubset}_{RS}$ and $\sqsubseteq_{RS}$ are denoted by $\approx_{RS}$ and $=_{RS}$ resp..
 It is easy to see that $\underset{\thicksim}{\sqsubset}_{RS}$ itself is a stable ready simulation relation and both $\underset{\thicksim}{\sqsubset}_{RS}$ and $\sqsubseteq_{RS}$ are pre-order.
\end{mydefn}

\subsection{The calculus $\text{CLL}_R$ and its operational semantics}

This subsection introduces the LLTS-oriented process calculus $\text{CLL}_R$ presented in \cite{Zhang14}.
Let $V_{AR}$ be an infinite set of variables.
The terms of $\text{CLL}_{R}$ can be given by the following BNF grammar
\[ t::= 0\;|\perp\;|\;(\alpha.t) \;|\; (t\Box t)\;|\;(t\wedge t)\;|\;(t\vee t)\;|\;(t\parallel_A t)\;|\;X\; | \;\langle Z|E \rangle \]
where $X \in V_{AR}$, $\alpha\in Act_\tau$, $A\subseteq Act$ and recursive specification
$E = E(V)$ with $V \subseteq V_{AR}$ is a set of equations $\{X = t| X \in V\}$ and $Z$ is a variable in $V$ that acts as the initial variable.

Most of these operators are from CCS \cite{Milner89} and CSP \cite{Hoare85}:
0 is the process capable of doing no action;
$\alpha.t$ is action prefixing;
$\Box$ is non-deterministic external choice;
$\parallel_A $ is a CSP-style parallel composition.
$\bot$ represents an inconsistent process with empty behavior.
$\vee$ and $\wedge$ are logical operators, which are intended for describing logical combinations of processes.

For any term $\langle Z|E \rangle$ with $E=E(V)$, each variable in  $V$ is bound with scope $E$.
This induces the notion of free occurrence of variable, bound (and free) variables and $\alpha$-equivalence as usual.
A term $t$ is a \emph{process} if it is closed, that is, it contains no free variable.
The set of all processes is denoted by $T(\Sigma_{\text{CLL}_R} )$.
Unless noted otherwise we use $p,q,r$ to represent processes.
Throughout this paper, as usual, we assume that recursive variables are distinct from each other and no recursive variable has free occurrence; moreover we don't distinguish between $\alpha$-equivalent terms and use $\equiv$ for both syntactical identical and $\alpha$-equivalence.
In the sequel, we often denote $\langle X|\{X=t_X\}\rangle$ briefly by $\langle X|X=t_X \rangle$.

 For any recursive specification $E(V)$ and term $t$, the term $\langle t|E \rangle$ is obtained from $t$ by simultaneously replacing all free occurrences of each $X(\in V)$ by $\langle X|E \rangle$, that is,  $\langle t|E \rangle \equiv t\{\langle X|E \rangle/X: X \in V\}$.
For example, consider $t \equiv X \Box a.\langle Y | Y = X \ \Box Y \rangle$ and $E(\{X\})=\{X=t_X\}$ then $\langle t| E\rangle \equiv \langle X|X =t_X\rangle \Box a.\langle Y | Y = \langle X|X=t_X\rangle \Box Y \rangle$.
In particular, for any $E(V)$ and $t \equiv X$, $\langle t|E \rangle \equiv \langle X|E\rangle$ whenever $X \in V$ and $\langle t|E \rangle \equiv X$ if $X \notin V$.

A context $C_{\widetilde{X}}$ is a term whose free variables are in some $n$-tuple distinct variables $\widetilde{X}=(X_1,...,X_n)$ with $n \geq 0$.
 Given $\widetilde{p}=(p_1,\dots,p_n)$, the term $C_{\widetilde{X}}\{p_1/X_1,...,p_n/X_n\}$ ($C_{\widetilde{X}}\{\widetilde{p}/\widetilde{X}\}$ for short) is obtained from $C_{\widetilde{X}}$ by replacing $X_i$ by $p_i$ for each $i \leq n$ simultaneously.
  A context $C_{\widetilde{X}}$ is stable if $C_{\widetilde{X}}\{\widetilde{0}/\widetilde{X}\} \not\stackrel{\tau}{\rightarrow} $.

An occurrence of $X$ in $t$ is strongly (or, weakly) guarded if such occurrence is within some subexpression $a.t_1$ with $a \in Act$ ($\tau.t_1$ or $t_1 \vee t_2$ resp.).
A variable $X$ is strongly (or, weakly) guarded in $t$ if each occurrence of $X$ is strongly (weakly resp.) guarded.
A recursive specification $E(V)$ is guarded if for each $X \in V$ and $Z = t_Z \in E(V)$, each occurrence of $X$ in $t_Z$ is (weakly or strongly) guarded.
 As usual, we assume that all recursive specifications considered in the remainder of this paper are guarded.
SOS rules of $\text{CLL}_R$ are listed in Table~\ref{Ta:OPERATIONAL_RULES}, where $a \in Act$, $\alpha \in Act_{\tau}$ and $A \subseteq Act$.
All rules are divided into two parts:

Operational rules specify behaviours of processes.
Negative premises in Rules $Ra_2$, $Ra_3$, $Ra_{13}$ and $Ra_{14}$ give $\tau$-transition precedence over visible transitions, which guarantees that the transition model of $\text{CLL}_{R}$ is $\tau$-pure.
Rules $Ra_9$ and $Ra_{10}$ illustrate that the operational aspect of $t_1\vee t_2$ is same as internal choice in usual process calculus.
Rule $Ra_6$ reflects that conjunction operator is a synchronous product for visible transitions.
The operational rules of the other operators are as usual.

Predicate rules specify  the inconsistency predicate $F$.
Rule $Rp_1$ says that $\bot$ is inconsistent.
Hence $\bot$ cannot be implemented.
While $0$ is consistent and implementable.
Thus $0$ and $\bot$ represent different processes.
Rule $Rp_3$ reflects that if both two disjunctive parts are inconsistent then so is the disjunction.
Rules $Rp_4-Rp_9$ describe the system design strategy that if one part is inconsistent, then so is the whole composition.
Rules $Rp_{10}$ and $Rp_{11}$ reveal that a stable conjunction is inconsistent whenever its conjuncts have distinct ready sets.
Rules $Rp_{13}$ and $Rp_{15}$ 
are used to capture (LTS2) in Def.~\ref{D:LLTS}.
Intuitively, these two rules say that if all stable $\tau$-descendants of $z$  are inconsistent, then $z$ itself is inconsistent.

\begin{table}[ht]
\rule{\textwidth}{0.5pt}
\noindent \textbf{Operational rules}\\
    $\begin{array}{lll}
    \displaystyle  \quad  Ra_1\frac{-}{\alpha.x_1 \stackrel{\alpha}{\rightarrow} x_1}  &
    \displaystyle  \quad  Ra_2\frac{x_1 \stackrel{a}{\rightarrow} y_1, x_2 \not \stackrel{\tau}{\rightarrow}}{x_1 \Box x_2 \stackrel{a}{\rightarrow} y_1} &
    \displaystyle  \;\;  Ra_3\frac{x_1 \not\stackrel{\tau}{\rightarrow} , x_2 \stackrel{a}{\rightarrow} y_2 }{x_1 \Box x_2 \stackrel{a}{\rightarrow} y_2} \\
    \displaystyle  \quad Ra_4\frac{x_1 \stackrel{\tau}{\rightarrow} y_1}{x_1 \Box x_2 \stackrel{\tau}{\rightarrow} y_1 \Box x_2}&
    \displaystyle   \quad Ra_5\frac{x_2 \stackrel{\tau}{\rightarrow} y_2}{x_1 \Box x_2 \stackrel{\tau}{\rightarrow} x_1 \Box y_2}&
    \displaystyle   \;\; Ra_6\frac{x_1 \stackrel{a}{\rightarrow} y_1, x_2 \stackrel{a}{\rightarrow}y_2}{x_1 \wedge x_2 \stackrel{a}{\rightarrow} y_1 \wedge y_2}\\
    \displaystyle   \quad Ra_7\frac{x_1 \stackrel{\tau}{\rightarrow} y_1}{x_1 \wedge x_2 \stackrel{\tau}{\rightarrow} y_1 \wedge x_2} &
    \displaystyle   \quad Ra_8\frac{x_2 \stackrel{\tau}{\rightarrow} y_2}{x_1 \wedge x_2 \stackrel{\tau}{\rightarrow} x_1 \wedge y_2}&
    \end{array}$

    $\begin{array}{ll}
    \displaystyle  \quad Ra_9\frac{-}{x_1 \vee x_2 \stackrel{\tau}{\rightarrow} x_1}
    &
    \displaystyle  \quad Ra_{10}\frac{-}{x_1 \vee x_2 \stackrel{\tau}{\rightarrow} x_2} \\
    \displaystyle  \quad Ra_{11}\frac{x_1 \stackrel{\tau}{\rightarrow} y_1}{x_1 \parallel_A x_2 \stackrel{\tau}{\rightarrow} y_1\parallel_A x_2} &
    \displaystyle  \quad Ra_{12}\frac{x_2 \stackrel{\tau}{\rightarrow} y_2}{x_1 \parallel_A x_2 \stackrel{\tau}{\rightarrow} x_1 \parallel_A y_2} \\
    \displaystyle  \quad Ra_{13}\frac{x_1 \stackrel{a}{\rightarrow} y_1 , x_2 \not \stackrel{\tau}{\rightarrow} }{x_1 \parallel_A x_2 \stackrel{a}{\rightarrow} y_1 \parallel_A x_2}(a\notin A)&
    \displaystyle  \quad Ra_{14}\frac{x_1 \not\stackrel{\tau}{\rightarrow} , x_2 \stackrel{a}{\rightarrow} y_2 }{x_1 \parallel_A x_2 \stackrel{a}{\rightarrow} x_1 \parallel_A y_2}(a\notin A)\\
    \displaystyle  \quad Ra_{15}\frac{x_1 \stackrel{a}{\rightarrow} y_1, x_2 \stackrel{a}{\rightarrow}y_2}{x_1\parallel_A x_2 \stackrel{a}{\rightarrow} y_1 \parallel_A y_2} (a\in A)&
    \displaystyle  \quad Ra_{16}\frac{\langle t_X| E \rangle  \stackrel{\alpha}{\rightarrow} y}{\langle X|E \rangle \stackrel{\alpha}{\rightarrow} y}(X=t_X \in E)
     \end{array}$

$\;$\\

\noindent \textbf{Predicative rules} \\
$\begin{array}{lll}
       \displaystyle \qquad Rp_1\frac{-}{\bot F}&
       \displaystyle \qquad\qquad Rp_2\frac{x_1 F}{\alpha .x_1 F} &
       \displaystyle \qquad\qquad Rp_3\frac{x_1 F, x_2 F}{x_1\vee x_2 F}\\
       \displaystyle  \qquad Rp_4\frac{x_1 F}{x_1\Box x_2 F}&
       \displaystyle \qquad\qquad Rp_5\frac{x_2 F}{x_1\Box x_2 F}&
       \displaystyle \qquad\qquad Rp_6\frac{x_1 F}{x_1\parallel_A x_2 F}\\
       \displaystyle \qquad Rp_7\frac{x_2 F}{x_1\parallel_A x_2 F}&
       \displaystyle \qquad\qquad Rp_8\frac{x_1 F}{x_1\wedge x_2 F}&
       \displaystyle \qquad\qquad Rp_9\frac{x_2 F}{x_1\wedge x_2 F}
    \end{array}$

    $\begin{array}{ll}
       \displaystyle \quad Rp_{10}\frac{x_1 \stackrel{a}{\rightarrow} y_1, x_2 \not\stackrel{a}{\rightarrow}, x_1 \wedge x_2 \not\stackrel{\tau}{\rightarrow}}{x_1 \wedge x_2 F}&
       \displaystyle \quad  Rp_{11}\frac{x_1 \not\stackrel{a}{\rightarrow} , x_2 \stackrel{a}{\rightarrow} y_2, x_1 \wedge x_2 \not\stackrel{\tau}{\rightarrow}}{x_1 \wedge x_2 F}\\
       \displaystyle \quad Rp_{12}\frac{x_1 \wedge x_2 \stackrel{\alpha}{\rightarrow} z, \{yF:x_1 \wedge x_2 \stackrel{\alpha}{\rightarrow}y\}}{x_1 \wedge x_2 F} &
       \displaystyle \quad Rp_{13}\frac{\{yF:x_1 \wedge x_2 \stackrel{\epsilon}{\Rightarrow}|y\}}{x_1 \wedge x_2 F} \\
       \displaystyle \quad Rp_{14}\frac{\langle t_X|E \rangle F}{\langle X|E \rangle F}(X = t_X \in E) &
       \displaystyle \quad Rp_{15}\frac{\{yF:\langle X|E \rangle \stackrel{\epsilon}{\Rightarrow}|y\}}{\langle X|E \rangle F}
    \end{array}
    $
\rule{\textwidth}{0.5pt}
\caption{SOS rules of $\text{CLL}_R$\label{Ta:OPERATIONAL_RULES}}
\end{table}

 It has been shown that $\text{CLL}_R$ has the unique stable transition model $M_{\text{CLL}_R}$ \cite{Zhang14},
 which exactly consists of all positive literals of the form $t \stackrel{\alpha}{\rightarrow}t'$ or $tF$ that are provable in $Strip(\text{CLL}_{R},M_{\text{CLL}_{R}})$.
 Here $Strip(\text{CLL}_{R},M_{\text{CLL}_{R}})$ is the stripped version \cite{Bol96} of $\text{CLL}_{R}$ w.r.t $M_{\text{CLL}_{R}}$.
 Each rule in $Strip(\text{CLL}_{R},M_{\text{CLL}_{R}})$ is of the form $\frac{pprem(r)}{conc(r)}$ for some ground instance  $r$ of rules in $\text{CLL}_R$ such that $M_{\text{CLL}_{R}}\models nprem(r)$, where  $nprem(r)$ (or, $pprem(r)$) is the set of negative (positive resp.) premises of $r$, $conc(r)$ is the conclusion of $r$ and $M_{\text{CLL}_{R}}\models nprem(r)$ means that for each $t\not\stackrel{\alpha}{\rightarrow} \in nprem(r)$, $t\stackrel{\alpha}{\rightarrow}s \notin M_{\text{CLL}_{R}}$ for any $s \in T(\Sigma_{\text{CLL}_{R}})$.

   The LTS associated with $\text{CLL}_{R}$, in symbols $LTS(\text{CLL}_{R})$, is the quadruple
    $(T(\Sigma_{\text{CLL}_{R}}),Act_{\tau},\rightarrow_{\text{CLL}_{R}},F_{\text{CLL}_{R}})$, where
    $p \stackrel{\alpha}{\rightarrow}_{\text{CLL}_{R}} p'$ iff $p\stackrel{\alpha}{\rightarrow} p' \in M_{\text{CLL}_{R}}$, and
    $p\in F_{\text{CLL}_{R}}$ iff $pF \in M_{\text{CLL}_{R}}$.
    Therefore $p \stackrel{\alpha}{\rightarrow}_{\text{CLL}_{R}} p'$ (or, $p \in F_{\text{CLL}_R}$) iff  $Strip( \text{CLL}_{R}, M_{\text{CLL}_{R}}) \vdash p\stackrel{\alpha}{\rightarrow}p'$ ($pF$ resp.) for any $p$, $p'$ and $\alpha \in Act_{\tau}$.
    For simplification, in the following we omit the subscripts in $\stackrel{\alpha}{\rightarrow}_{\text{CLL}_{R}}$ and $F_{\text{CLL}_{R}}$.

We end this section by quoting some results from \cite{Zhang14}.

\begin{lemma}\label{L:F_NORMAL}
Let $p$ and $q$ be any two processes. Then

    \noindent (1) $p \vee q \in F $  iff $p,q \in F $;\\
    \noindent (2)  $\alpha.p \in F $ iff $p \in F $ for each $\alpha \in Act_{\tau}$;\\
    \noindent (3)  $p \odot q \in F $  iff either $p \in F $ or $q \in F $  with $\odot \in \{\Box, \parallel_A\}$;\\
    \noindent (4) $p \in F $ or $q \in F $ implies $p \wedge q \in F $;\\
    \noindent (5) $0 \notin F $ and $\bot \in F $;\\
    \noindent (6) $\langle X|E \rangle \in F$ iff $\langle t_X|E\rangle \in F$ for each $X$ with $X=t_X \in E$.
\end{lemma}

 \begin{theorem}\label{L:LLTS}
    $LTS({\text{CLL}_{R}})$ is a $\tau$-pure LLTS. Moreover if $p\in F$ and $\tau\in \mathcal{I}(p)$ then $\forall q(p\stackrel{\tau}{\rightarrow}q \;\text{implies}\; q \in F)$.
\end{theorem}

\begin{theorem}[precongruence]\label{L:precongruence}
  If $p \sqsubseteq_{RS} q$ then $C_X\{p/X\} \sqsubseteq_{RS} C_X\{q/X\}$.
\end{theorem}

\section{More on solutions of equations in $\text{CLL}_R$}
In \cite{Zhang14}, the following theorem has been obtained.\\

\noindent \textbf{Theorem} (Unique solution). For any $p,q \notin F $ and $t_X$ where $X$ is strongly guarded and does not occur in the scope of any conjunction, if $p =_{RS} t_X\{p/X\}$ and $q =_{RS} t_X\{q/X\}$ then $p =_{RS} q$.
Moreover $\langle X | X=t_X \rangle$ is the unique consistent solution (modulo $=_{RS}$) of the equation $X =_{RS} t_X$ whenever  consistent solutions exist.\\

As we know, temporal operators could be described in equational style, represented by fixpoint of some equations \cite{ModelChecking}.
Such style requires us to remove the special requirement (i.e. $X$ does not occur in the scope of any conjunction) occurring in Theorem Unique Solution.
In the following, we give a negative answer for this removement by providing a counterexample:

\begin{observation}
  Consider the equation $X=t_X$ where $t_X \equiv (\langle Y |Y =a.Y\rangle \wedge a.X) \vee (\langle Z |Z =b.Z\rangle \wedge b.X)$.
  In the following, we show that $\langle X|X=a.X \rangle$ is a consistent solution of this equation.
  First we show that $\langle X|X=a.X \rangle \notin F$. Contrarily, assume that $\langle X|X=a.X \rangle \in F$. Then the last rule applied in the proof tree of $Strip(\text{CLL}_R,M_{\text{CLL}_R}) \vdash \langle X| X = a.X \rangle F$ is
  \[\frac{a.\langle X| X = a.X  \rangle F}{\langle X| X =a.X \rangle F}\;\text{or}\;\frac{\{rF:\langle X| X =a.X\rangle  \stackrel{\epsilon}{\Rightarrow}|r\}}{\langle X| X =a.X \rangle F}.\]
  It is not difficult to see that every proof tree of $\langle X| X = a.X \rangle F$ has proper subtree with root $\langle X| X = a.X \rangle F$, this contradicts the well-foundedness of proof tree, as desired.
  Second we show that $\langle X|X=a.X \rangle$ indeed is a solution of $X=_{RS}t_X$. Clearly, due to Rules $Rp_{10}$ and $Rp_{11}$, $\langle Z |Z =b.Z\rangle \wedge \langle X|X=a.X \rangle \in F$, which is the unique $b$-derivative of $\langle Z |Z =b.Z\rangle \wedge b.\langle X|X=a.X \rangle$.
  Hence  $\langle Z |Z =b.Z\rangle \wedge b.\langle X|X=a.X \rangle \in F$ by Condition (LTS1) in Def.~\ref{D:LLTS} and Theorem~\ref{L:LLTS}.
  Moreover we also have $\langle X|X=a.X \rangle =_{RS} \langle Y|Y=a.Y \rangle \wedge a.\langle X|X=a.X\rangle$.
  Therefore $\langle X|X=a.X \rangle =_{RS} t_X\{\langle X|X=a.X \rangle /X\}$.
  Similarly, $\langle X|X=b.X \rangle$ is another consistent solution. However, $\langle X|X=a.X \rangle \not=_{RS}\langle X|X=b.X \rangle$.
\end{observation}

  In the remainder of this section, we intend to show that the recursive process $\langle X|X=t \rangle$ captures the extreme solution of the equation $X=t$.
  To this end, a number of results in \cite{Zhang14} are listed below.

  \begin{lemma}\label{L:ONE_ACTION_TAU}
    If $C_X\{p/X\} \stackrel{\tau}{\rightarrow} r$ then

    \noindent(1) either there exists $C_X'$ such that $r\equiv C_X'\{p/X\}$ and $C_X\{q/X\} \stackrel{\tau}{\rightarrow}C_X'\{q/X\}$ for any $q$,

    \noindent (2) or there exist $C_{X,Z}'$ and $p'$ such that $p \action{\tau} p'$, $r \equiv C_{X,Z}'\{p/X,p'/X\}$ and $C_X\{q/X\} \action{\tau} C_{X,Z}'\{q/X,q'/Z\}$ for any $q \action{\tau} q'$.
  \end{lemma}

  \begin{lemma}\label{L:ONE_ACTION_VISIBLE}
    Let $a \in Act$. If $C_X\{p/X\} \action{a} r$ then there exits $C_{X,\widetilde{Y}}'$ such that

    \noindent (1) $r \equiv C_{X,\widetilde{Y}}'\{p/X,\widetilde{p_Y'}/\widetilde{Y}\}$ for some $\widetilde{p_Y'}$ with $p \action{a} p_Y'$ for each $Y \in \widetilde{Y}$, and

    \noindent (2) if $C_X\{q/X\}$ is stable and for each $Y \in \widetilde{Y}$, $q \action{a} q_Y'$, then $C_X\{q/X\}\action{a} C_{X,\widetilde{Y}}'\{q/X,\widetilde{q_Y'}/\widetilde{Y}\}$.
  \end{lemma}

  \begin{lemma}\label{L:ONE_ACTION_VISIBLE_GUARDED}
    Let $X$ be guarded in $C_X$.
    If $C_X\{p/X\} \action{\alpha} r$ then there exists $B_X$ such that $r \equiv B_X\{p/X\}$ and $C_X\{q/X\} \action{\alpha}B_X\{q/X\}$ for any $q$.
  \end{lemma}

  \begin{lemma}\label{L:MULTI_TAU_GF_STABLE}
    If $C_X\{p/X\}\sactions{\epsilon}r$ then there exist stable $C_{X,\widetilde{Y}}'$ and stable $p_Y'$ for each $Y \in \widetilde{Y}$ such that
    \noindent (1) $p \sactions{\tau}p_Y'$ for each $Y \in \widetilde{Y}$ and $r\equiv C_{X,\widetilde{Y}}'\{p/X,\widetilde{p_Y'}/\widetilde{Y}\}$;
    \noindent (2) for any $q$ such that $q \action{\tau}$ iff $p\action{\tau}$, if $q \sactions{\tau}q_Y'$ for each $Y \in \widetilde{Y}$ then $C_X\{q/X\}\sactions{\epsilon}C_{X,\widetilde{Y}}'\{q/X,\widetilde{q_Y'}/\widetilde{Y}\}$;
    \noindent (3) if $X$ is strongly guarded in $C_X$ then so it is in $C_{X,\widetilde{Y}}'$ and $\widetilde{Y}= \emptyset$.
  \end{lemma}

  Before giving the main result of this section, we prove a lemma concerning $F$-predicate.

  \begin{lemma}\label{L:precongruence_F}
    If $X$ is strongly guarded in $t_X$ and $p\sqsubseteq_{RS}t_X\{p/X\}$ then for any $C_Y$, $C_Y\{t_X\{p/X\}/Y\} \notin F$ implies $C_Y\{\langle X|X=t_X \rangle/Y\} \notin F$.
  \end{lemma}
  \begin{proof}
  Clearly, by Lemmas~\ref{L:ONE_ACTION_TAU}, \ref{L:ONE_ACTION_VISIBLE} and \ref{L:ONE_ACTION_VISIBLE_GUARDED}, we get
  \[C_Y\{t_X\{p/X\}/Y\}\action{\alpha}\;\text{iff}\;C_Y\{\langle X|X=t_X \rangle /Y\}\action{\alpha}\;\text{for any}\;C_Y. \tag{\ref{L:precongruence_F}.1}\]
    Set $\Omega \triangleq \{B_Y\{\langle X|X=t_X \rangle/Y\}:B_Y\{t_X\{p/X\}/Y\} \notin F\}$.
    Clearly, it suffice to prove that $F \cap \Omega = \emptyset$.
    Conversely, suppose that $F \cap \Omega \neq \emptyset$.
    Due to the well-foundedness of proof trees, to complete the proof, it is sufficient to show that, for each $C_Y\{\langle X|X=t_X \rangle /Y\} \in \Omega$, any proof tree for $Strip(\text{CLL},M_{\text{CLL}_R})\vdash C_Y\{\langle X|X=t_X \rangle /Y\}F$ has a proper subtree with root $sF$ for some $s \in \Omega$.
    We shall prove this as follows.
    Let $\mathcal T$ be any proof tree of $C_Y\{\langle X|X=t_X \rangle /Y\}F$.
    It is a routine case analysis based on the last rule applied in $\mathcal T$.
    We treat only non-trivial three cases and leave the others to the reader.

    \noindent Case 1. $C_Y \equiv Y$.

    Then $C_Y\{\langle X|X=t_X \rangle /Y\} \equiv \langle X|X=t_X \rangle$. So the last rule applied in $\mathcal T$ is $\frac{\langle t_X|X=t_X \rangle F}{\langle X|X=t_X \rangle F}$ or $\frac{\{rF:\langle X|X=t_X \rangle \sactions{\epsilon}r\}}{\langle X|X=t_X \rangle F}$.

    For the former, since  $C_Y\{t_X\{p/X\}/Y\} \equiv t_X\{p/X\} \notin F$ and $p\sqsubseteq_{RS}t_X\{p/X\}$, we have $t_X\{t_X\{p/X\}/X\} \notin F$ due to Theorem~\ref{L:precongruence}.
    Hence $\langle t_X|X=t_X\rangle \equiv t_X\{\langle X|X=t_X \rangle/X\} \in \Omega$.
    For the latter, we treat the non-trivial subcase that $\langle X|X=t_X \rangle \action{\tau}$.
    Since $t_X\{p/X\} \notin F$, $t_X\{p/X\}\fsactions{\epsilon}s$ for some $s$.
    For this transition, since $X$ is strongly guarded in $t_X$, by Lemma~\ref{L:MULTI_TAU_GF_STABLE}, there exists a stable $t_X'$ with strongly guarded $X$ such that $s \equiv t_X'\{p/X\}$ and $t_X\{\langle X|X=t_X\rangle\} \actions{\epsilon} t_X'\{\langle X|X=t_X \rangle /X\}$.
    Further, by Lemma~\ref{L:ONE_ACTION_VISIBLE_GUARDED}, $\langle X|X=t_X\rangle \sactions{\tau} t_X'\{\langle X|X=t_X \rangle/X\}$ due to $\langle X|X=t_X\rangle \action{\tau}$.
    Moreover $t_X'\{t_X\{p/X\}/X\} \notin F$ because of $s\equiv t_X'\{p/X\} \notin F$ and $p \RS t_X\{p/X\}$.
    Hence $t_X'\{\langle X|X=t_X \rangle /X\} \in \Omega$, as desired.\\

    \noindent Case 2. $C_Y \equiv \langle Z|E \rangle$.

    The last rule applied in $\mathcal T$ is one of following two cases: $\frac{\langle t_Z|E \rangle \{\langle X|X=t_X \rangle /Y\} F}{\langle Z|E \rangle \{\langle X|X=t_X \rangle /Y\} F}$ or $\frac{\{rF:\langle Z|E\rangle \{\langle X|X=t_X \rangle /Y\} \sactions{\epsilon}r\}}{\langle Z|E \rangle \{\langle X|X=t_X \rangle /Y\} F}$.

    By Lemma~\ref{L:F_NORMAL}(6), the former is easy to handle and omitted. Next we treat the latter.
    Since $C_Y\{t_X\{p/X\}/Y\} \notin F$, $C_Y\{t_X\{p/X\}/Y\}  \fsactions{\epsilon}s$ for some $s$.
    For this transition, by Lemma~\ref{L:MULTI_TAU_GF_STABLE}, there exist stable $C_{Y,\widetilde{W}}'$ and $\widetilde{s_W'}$ such that $s\equiv C_{Y,\widetilde{W}}'\{t_X\{p/X\},\widetilde{s_W'}/\widetilde{W}\}$ and $t_X\{p/X\} \sactions{\tau} s_W'$ for each $W\in \widetilde{W}$.
    Further, for each $t_X\{p/X\} \sactions{\tau}s_W'$, there exists stable $t_X'^W$ with strongly guarded $X$ such that $s_W' \equiv t_X'^W\{p/X\}$ and $t_X\{\langle X|X=t_X \rangle /X\} \actions{\tau} t_X'^W\{\langle X|X=t_X\rangle/X\}$.
    So, by Lemma~\ref{L:ONE_ACTION_VISIBLE_GUARDED}, $\langle X|X=t_X \rangle   \sactions{\tau} t_X'^W\{\langle X|X=t_X\rangle/X\}$ for each $W \in \widetilde{W}$ and hence $C_Y\{\langle X|X=t_X\rangle/Y\}  \sactions{\epsilon} C_{Y,\widetilde{W}}'\{\langle X|X=t_X \rangle/Y,\widetilde{t_X'^W\{\langle X|X=t_X \rangle /X\}}/\widetilde{W}\}\equiv u$.
    Since $s\equiv C_{Y,\widetilde{W}}'\{t_X\{p/X\}/Y,\widetilde{t_X'^W\{p/X\}}/\widetilde{W}\} \notin F$ and $p \RS t_X\{p/X\}$, we get $C_{Y,\widetilde{W}}'\{t_X\{p/X\}/Y,\widetilde{t_X'^W\{t_X\{p/X\}/X\}}/\widetilde{W}\} \notin F$, which implies $u \in \Omega$, as desired.\\

 \noindent   Case 3. $C_Y \equiv B_Y \wedge D_Y$.

 We split the argument into the following four subcases.

 \noindent Case 3.1. $\frac{B_Y\{\langle X|X=t_X \rangle/Y\}F}{C_Y\{\langle X|X=t_X \rangle/Y\}F}$.

 Since $C_Y\{t_X\{p/X\}/Y\} \notin F$, $B_Y\{t_X\{p/X\}/Y\} \notin F$ by Lemma~\ref{L:F_NORMAL}.
 So, $B_Y\{\langle X|X=t_X \rangle/Y\} \notin F$, as desired.\\

 \noindent Case 3.2. $\frac{B_Y\{\langle X|X=t_X \rangle/Y\} \action{a}}{C_Y\{\langle X|X=t_X \rangle/Y\}F}$ with $D_Y\{\langle X|X=t_X \rangle/Y\} \not\action{a}$ and $C_Y\{\langle X|X=t_X \rangle/Y\}\not\action{\tau}$.

 By (\ref{L:precongruence_F}.1), a contradiction arises due to $C_Y\{t_X\{p/X\}/Y\} \notin F$.\\

 \noindent Case 3.3. $\frac{\{rF:C_Y\{\langle X|X=t_X \rangle/Y\} \sactions{\epsilon}r\}}{C_Y\{\langle X|X=t_X \rangle/Y\}F}$.

 Similar to the second case of Case~2, omitted.\\

 \noindent Case 3.4. $\frac{C_Y\{\langle X|X=t_X \rangle/Y\} \action{\alpha}r', \{rF:C_Y\{\langle X|X=t_X \rangle/Y\} \action{\alpha}r\}}{C_Y\{\langle X|X=t_X \rangle/Y\}F}$.

 Then $C_Y\{\langle X|X=t_X \rangle/Y\} \action{\alpha}r'$. Since $C_Y\{t_X\{p/X\}/Y\} \notin F$, by (\ref{L:precongruence_F}.1),
 \[C_Y\{t_X\{p/X\}/Y\} \faction{\alpha}s\;\text{for some}\;s. \tag{\ref{L:precongruence_F}.2}\]
 In the following, we treat two cases based on $\alpha$.

 \noindent Case~3.4.1. $\alpha=\tau$.

    For (\ref{L:precongruence_F}.2), by Lemma~\ref{L:ONE_ACTION_TAU}, either $s\equiv C_Y'\{t_X\{p/X\}/Y\}$ for some $C_Y'$ such that $C_Y\{q/Y\} \action{\tau} C_Y'\{q/Y\}$ for any $q$, or there exist $s'$ and $C_{Y,Z}'$ such that $s\equiv C_{Y,Z}'\{t_X\{p/X\}/Y,s'/Z\}$ and $t_X\{p/X\} \action{\tau} s'$.
    For the former, it is trivial.
    Next we treat the later.
    For $t_X\{p/X\} \action{\tau} s'$, since $X$ is strongly guarded in $t_X$, by Lemma~\ref{L:ONE_ACTION_VISIBLE_GUARDED}, there exists $t_X'$ such that $s'\equiv t_X'\{p/X\}$ and $t_X\{\langle X|X=t_X\rangle/X\} \action{\tau} t_X'\{\langle X|X=t_X\rangle/X\}$.
    Then $\langle X|X=t_X\rangle \action{\tau} t_X'\{\langle X|X=t_X\rangle/X\}$ and hence $C_Y\{\langle X|X=t_X \rangle /Y\}\action{\tau} C_{Y,Z}'\{\langle X|X=t_X \rangle /Y,t_X'\{\langle X|X=t_X\rangle/X\}/Z\}\equiv u$.
    Since $p \RS t_X\{p/X\}$ and $s\equiv C_{Y,Z}'\{t_X\{p/X\}/Y,t_X'\{p/X\}/Z\} \notin F$, we get $C_{Y,Z}'\{t_X\{p/X\}/Y,t_X'\{t_X\{p/X\}/X\}/Z\} \notin F$.
    Clearly, $u\in \Omega$, as desired.\\

 \noindent Case~3.4.2. $\alpha \in Act$.

 For (\ref{L:precongruence_F}.2), by Lemma~\ref{L:ONE_ACTION_VISIBLE}, $s \equiv C_{Y,\widetilde{Z}}'\{t_X\{p/X\},\widetilde{s_Z'}/\widetilde{Z}\}$ for some $C_{Y,\widetilde{Z}}'$ and $\widetilde{s_Z'}$ such that $t_X\{p/X\} \action{\alpha} s_Z'$ for each $Z \in \widetilde{Z}$.
 Since $X$ is strongly guarded in $t_X$, for each $t_X\{p/X\} \action{\alpha} s_Z'$, by Lemma~\ref{L:ONE_ACTION_VISIBLE_GUARDED}, there exists $t_X'^Z$ such that $s_Z' \equiv t_X'^Z\{p/X\}$ and $t_X\{\langle X|X=t_X\rangle/X\} \action{\alpha} t_X'^Z\{\langle X|X=t_X \rangle/X\}$.
 Then $\langle X|X=t_X\rangle \action{\alpha} t_X'^Z\{\langle X|X=t_X \rangle/X\}$ for each $Z \in \widetilde{Z}$ and hence $C_Y\{\langle X|X=t_X \rangle/Y\} \action{\alpha} C_{Y,\widetilde{Z}}'\{\langle X|X=t_X \rangle /Y, \widetilde{t_X'^Z\{\langle X|X=t_X \rangle /X \}}/\widetilde{Z}\} \equiv u$ by (\ref{L:precongruence_F}.1).
 Since $p \RS t_X\{p/X\}$ and $s \equiv C_{Y,\widetilde{Z}}'\{t_X\{p/X\},\widetilde{t_X'^Z\{p/X\}}/\widetilde{Z}\} \notin F$, by Theorem~\ref{L:precongruence}, we get $C_{Y,\widetilde{Z}}'\{t_X\{p/X\},\widetilde{t_X'^Z\{t_X\{p/X\}/X\}}/\widetilde{Z}\} \notin F$.
 Clearly, $u \in \Omega$, as desired.
  \end{proof}

  Next we recall an equivalent formulation of $\sqsubseteq_{RS}$ and an up-to technique.

  \begin{mydefn}\label{D:ALT_RS}
A relation ${\mathcal R} \subseteq T(\Sigma_{\text{CLL}_R})\times T(\Sigma_{\text{CLL}_R})$ is an alternative ready simulation relation, if for any $(p,q) \in {\mathcal R}$ and $a \in Act $\\
\textbf{(RSi)} $p \stackrel{\epsilon}{\Rightarrow}_F|p'$ implies $\exists q'.q \stackrel{\epsilon}{\Rightarrow}_F|q'\; \textrm{and}\;(p',q') \in {\mathcal R}$;\\
\textbf{(RSiii)} $p \stackrel{a}{\Rightarrow}_F|p'$ and $p,q$ stable implies $\exists q'.q \stackrel{a}{\Rightarrow}_F|q'\; \textrm{and}\;(p',q') \in {\mathcal R}$;\\
\textbf{(RSiv)} $p\notin F $ and $p,q$ stable implies ${\mathcal I}(p)={\mathcal I}(q)$.

We write $p \sqsubseteq_{ALT} q$ if there exists an alternative ready simulation relation $\mathcal R$ with $(p,q) \in \mathcal R$.
\end{mydefn}
\begin{mydefn}[ALT up to $\underset{\thicksim}{\sqsubset}_{RS}$]\label{D:ALT_RS_UPTO}
A relation ${\mathcal R} \subseteq T(\Sigma_{\text{CLL}_R})\times T(\Sigma_{\text{CLL}_R})$ is an alternative ready simulation relation up to $\underset{\thicksim}{\sqsubset}_{RS}$, if for any $(p,q) \in {\mathcal R}$ and $a \in Act $\\
\textbf{(ALT-upto-1)} $p \fsactions{\epsilon}p'$ implies $\exists q'.q \fsactions{\epsilon}q'$ and $p' \SRS {\mathcal R}\SRS q'$;\\
\textbf{(ALT-upto-2)} $p \fsactions{a}p'$ and $p,q$ stable implies $\exists q'.q \fsactions{a}q'$ and $p'\SRS{\mathcal R}\SRS q'$;\\
\textbf{(ALT-upto-3)} $p\notin F $ and $p,q$ stable implies ${\mathcal I}(p)={\mathcal I}(q)$.
\end{mydefn}

It has been proved that $\RS = \sqsubseteq_{ALT}$ \cite{Luttgen10} and if $\mathcal R$ is an alternative ready simulation relation up to $\SRS$, then ${\mathcal R} \subseteq \RS$ \cite{Zhang14}. With these results, we could prove the next lemma.

\begin{lemma}
  Let $X$ be strongly guarded in $t_X$. If $p \RS t_X\{p/X\}$ then $t_X\{p/X\} \RS \langle X|X=t_X \rangle$.
\end{lemma}
\begin{proof}
  Set ${\mathcal R} \triangleq \{(B_Y\{t_X\{p/X\}/Y\},B_Y\{\langle X|X=t_X \rangle/Y\})\}$. It is sufficient to prove that $\mathcal R$ is an alternative ready simulation relation up to $\SRS$.
  Let $(C_Y\{t_X\{p/X\}/Y\},C_Y\{\langle X|X=t_X \rangle/Y\}) \in \mathcal R$.
  By Lemma~\ref{L:ONE_ACTION_TAU}, \ref{L:ONE_ACTION_VISIBLE} and \ref{L:ONE_ACTION_VISIBLE_GUARDED}, (ALT-upto-3) holds clearly.
  Next we handle the other two clauses.

  \textbf{(ALT-upto-1)} Assume $C_Y\{t_X\{p/X\}/Y\} \fsactions{\epsilon} s$. For this transition, by Lemma~\ref{L:MULTI_TAU_GF_STABLE}, $s \equiv C_{Y,\widetilde{Z}}'\{t_X\{p/X\}/Y,\widetilde{s_Z'}/\widetilde{Z}\}$ for some stable $C_{Y,\widetilde{Z}}'$ and $\widetilde{s_Z'}$ such that $t_X\{p/X\} \sactions{\tau} s_Z'$ \everyR{Z}.
  Further, for each $t_X\{p/X\} \sactions{\tau} s_Z'$, since $X$ is strongly guarded in $t_X$, there exists stable $t_X'^Z$  with strongly guarded $X$ such that $s_Z' \equiv t_X'^Z\{p/X\}$ and $t_X\{\langle X|X=t_X\rangle /X\} \sactions{\tau} t_X'^Z\{\langle X|X=t_X \rangle /X\}$.
  So $\langle X|X=t_X\rangle \sactions{\tau} t_X'^Z\{\langle X|X=t_X \rangle /X\}$ \everyR{Z} and hence $C_Y\{\langle X|X=t_X \rangle /Y\} \sactions{\epsilon} C_{Y,\widetilde{Z}}'\{\langle X|X=t_X \rangle /Y, \widetilde{t_X'^Z\{\langle X|X=t_X \rangle /X\}} /\widetilde{Z}\} \equiv u$.
  Since $s \equiv C_{Y,\widetilde{Z}}'\{t_X\{p/X\}/Y, \widetilde{t_X'^Z\{p /X\}} /\widetilde{Z}\} \notin F$ and $p \RS t_X\{p/X\}$, by Lemma~\ref{L:ONE_ACTION_VISIBLE_GUARDED} and Theorem~\ref{L:precongruence}, we obtain $s \SRS C_{Y,\widetilde{Z}}'\{t_X\{p/X\}/Y, \widetilde{t_X'^Z\{t_X\{p/X\} /X\}} /\widetilde{Z}\} \notin F$, which implies $u \notin F$ by Lemma~\ref{L:precongruence_F}.
  Clearly $C_Y\{\langle X|X=t_X\rangle/Y\} \fsactions{\epsilon} u$ by Lemma~\ref{L:LLTS}, and $s \SRS {\mathcal R} u$, as desired.

  \textbf{(ALT-upto-2)} Assume that $C_Y\{t_X\{p/X\}/Y\}$ and $C_Y\{\langle X|X=t_X\rangle /Y\}$ are stable and $C_Y\{t_X\{p/X\}/Y\} \fsactions{a} s$.
  Then $C_Y\{t_X\{p/X\}/Y\} \faction{a} r\fsactions{\epsilon} s$ for some $r$.
  For the $a$-transition, by Lemma~\ref{L:ONE_ACTION_VISIBLE}, $r \equiv C_{Y,\widetilde{Z}}'\{t_X\{p/X\}/Y,\widetilde{r_Z'}/\widetilde{Z}\}$ for some $C_{Y,\widetilde{Z}}'$ and $\widetilde{r_Z'}$ such that $t_X\{p/X\} \action{a} r_Z'$ \everyR{Z}.
  Since $X$ is strongly guarded in $t_X$, for each $t_X\{p/X\} \action{a} r_Z'$, by Lemma~\ref{L:ONE_ACTION_VISIBLE_GUARDED}, there exists $t_X'^Z$ such that $r_Z' \equiv t_X'^Z\{p/X\}$ and $t_X\{\langle X|X=t_X\rangle /X\} \action{a} t_X'^Z\{\langle X|X=t_X \rangle /X\}$.
  Then $ \langle X|X=t_X\rangle   \action{a} t_X'^Z\{\langle X|X=t_X \rangle /X\}$ \everyR{Z} and hence $C_Y\{\langle X|X=t_X \rangle/Y\} \action{a} C_{Y,\widetilde{Z}}'\{\langle X|X=t_X \rangle /Y, \widetilde{t_X'^Z\{\langle X|X=t_X \rangle /X\}}/\widetilde{Z} \} \equiv v$.
  Let $u \equiv C_{Y,\widetilde{Z}}'\{t_X\{p/X\} /Y, \widetilde{t_X'^Z\{t_X\{p/X\} /X\}}/\widetilde{Z} \}$.
  Since $p \RS t_X\{p/X\}$, by Theorem~\ref{L:precongruence}, we have $r \equiv C_{Y,\widetilde{Z}}'\{t_X\{p/X\} /Y, \widetilde{t_X'^Z\{p /X\}}/\widetilde{Z} \} \RS u$.
  Hence since $r \fsactions{\epsilon} s$, we have $u \fsactions{\epsilon} t$ and $s \SRS t$ for some $t$.
  Since $u {\mathcal R} v$, by (ALT-upto-1), $v \fsactions{\epsilon} t'$ for some $t'$ such that $t \SRS {\mathcal R} \SRS t'$.
  Therefore, by Lemma~\ref{L:precongruence_F}, $C_Y\{\langle X|X=t_X \rangle /Y\} \fsactions{a} t'$ and $s \SRS t \SRS {\mathcal R} \SRS t'$.
\end{proof}

Now with the previous lemma, it is not difficult to get

\begin{theorem}
  For any equation $X=_{RS}t_X$ such that $X$ is strongly guarded in $t_X$, if consistent solution exists then $\langle X|X=t_X\rangle$ is the greatest consistent solution.
\end{theorem}

\section{Encoding ACTL in $\text{CLL}_R$}

In \cite{Luttgen11}, {\they} introduce a fragment of action-based CTL \cite{Nicola90} (ACTL for short), embed it into LLTS and present the desired compatibility result between logical satisfaction and $\sqsubseteq_{RS}$.
In this section, we recall their ACTL and encode it in $\text{CLL}_R$ under the hypothesis that $Act$ is finite.

\begin{mydefn}
  The action-based CTL is defined by BNF:
  \[\phi ::= tt \mid ff \mid en(a) \mid dis(a) \mid \phi \vee \phi \mid \phi \wedge \phi \mid [a]\phi \mid \mathcal{A}\phi \mid \phi \mathcal{W} \phi\]
  where $a \in Act$.
  $T(\Sigma_{\text{ACTL}})$ denotes the set of all terms in ACTL.
\end{mydefn}
$en(a)$ and $dis(a)$ denote enabledness and disabledness of action $a$ resp.
$[a]$, $\mathcal{A}$ and $\mathcal{W}$ are usual \emph{next}, \emph{always} and \emph{weak until} operators.
For more motivations and intuitions about these operators, the reader may refer to \cite{Luttgen11}.

Before encoding formulas of ACTL in $\text{CLL}_R$, we introduce some useful notations.
Given $n$ terms $t_i(0 \leq i \leq n-1)$ in $T(\Sigma_{\text{CLL}_R})$, the general external choice $\underset{i<n}{\square}t_i$ and disjunction $\underset{i<n}{\bigvee}t_i$ are defined recursively as:

$\underset{i<0}\square t_i \triangleq 0,
   \underset{i<1}\square t_i \triangleq t_0,\;\text{and}\;
   \underset{i<k+1}\square t_i \triangleq (\underset{i<k}\square t_i) \Box t_k \;\text{for}\; k \geq 1;$

$\underset{i<1}\bigvee t_i \triangleq t_0,\;\text{and}\;
    \underset{i<k+1}\bigvee t_i \triangleq (\underset{i<k}\bigvee t_i) \vee t_k\; \text{for} \;k\geq 1.$

\noindent The general conjunction $\underset{i<n}{\bigwedge}t_i$ is defines similarly as disjunction.

Given a term $\phi$ in $T(\Sigma_{\text{ACTL}})$, the encoding of $\phi$, denoted by $\mathcal{E}(\phi)$, is defined as:

\noindent
 $\begin{array}{ll}
     \mathcal{E}(tt)   \triangleq   \langle X|X= \underset{A \subseteq Act}{\bigvee}\underset{a\in A}{\square}a.X \rangle
   &  \qquad  \mathcal{E}(ff)   \triangleq  \bot \\
      \mathcal{E}(en(a))   \triangleq   \underset{a \in A \subseteq Act}{\bigvee}\underset{b\in A}{\square}b.\mathcal{E}(tt)
   &  \qquad \mathcal{E}(dis(a))  \triangleq  \underset{a \notin A \subseteq Act}{\bigvee}\underset{b\in A}{\square}b.\mathcal{E}(tt)\\
      \mathcal{E}([a]\phi)   \triangleq   \lceil a\rceil(\mathcal{E}(\phi)) &
       \qquad \mathcal{E}(\phi_1 \vee \phi_2)  \triangleq  \mathcal{E}(\phi_1) \vee \mathcal{E}(\phi_2) \\
      \mathcal{E}(\phi_1 \wedge \phi_2)   \triangleq   \mathcal{E}(\phi_1) \wedge \mathcal{E}(\phi_2)
      & \qquad \mathcal{E}(\mathcal{A}\phi)   \triangleq  \langle X| X= \mathcal{E}(\phi) \wedge (\underset{a \in Act}{\bigwedge}\lceil a\rceil X ) \rangle
  \end{array}$

\noindent  $   \;\mathcal{E}(\phi_1 \mathcal{W} \phi_2)   \triangleq  \langle X| X= \mathcal{E}(\phi_2)  \vee ( \mathcal{E}(\phi_1) \wedge (\underset{a \in Act}{\bigwedge}\lceil a \rceil (X))\rangle$

\noindent where $   \lceil a\rceil \triangleq \lambda x. ( \underset{a\in A \subseteq Act}{\bigvee}((\underset{b \in A-\{a\}}{\square}b.\mathcal{E}(tt) )\Box a.x)) \vee (\underset{a\notin A \subseteq Act}{\bigvee}(\underset{b \in A}{\square}b.\mathcal{E}(tt)))$, intuitively, $\lceil a\rceil$ says ``along $a$-transition, it is necessary that \dots''.

Therefore, if we want to check a specification $p \in T(\Sigma_{\text{CLL}_R})$ satisfies some desired property $\phi \in T(\Sigma_{\text{ACTL}})$, we only check whether $p \RS {\mathcal E}(\phi)$ or $p \wedge {\mathcal E}(\phi)=_{RS} \bot$ holds.

\begin{theorem}
  $p\vDash \phi$ iff $p \RS \mathcal{E}(\phi)$.
\end{theorem}
\section{Conclusions and discussion}

  This paper works on LLTS-oriented process calculus $\text{CLL}_R$ furtherly.
  We show that for any given equation $X=_{RS}t$ such that $X$ is strongly guarded in $t$, $\langle X|X=t \rangle$ is the largest consistent solution w.r.t $\RS$ if consistent solutions exist.
  Moreover we also encode a temporal logic language ACTL in $\text{CLL}_R$.

  For further work, it is very interesting to study the structure of the solution space $\{p:p \RS t_X\{p/X\}\}$ if $X$ is strongly guarded in $t_X$.\\

\end{document}